\documentclass[journal]{IEEEtran}

\usepackage{mathtools}

\usepackage{amssymb}
\usepackage{graphics} 
\usepackage{epsfig} 
\usepackage{mathptmx} 
\usepackage{times} 
\usepackage{amsmath} 
\usepackage{amssymb}  
\usepackage{multicol}
\usepackage{amssymb, amsthm, amsfonts, mathrsfs,dsfont}
\usepackage{multirow}
\usepackage{caption}
\usepackage{subcaption}
\usepackage{graphicx}
\usepackage{epstopdf}
\usepackage{tikz-cd}
\usetikzlibrary{fit}

\usepackage{times}
\usepackage{amsmath}
\usepackage{amssymb}

\usepackage{algorithm}
\usepackage[noend]{algpseudocode}
\selectfont

\newtheorem{theorem}{Theorem}
\newtheorem{lemma}{Lemma}
\newtheorem{itassumption}{Assumption}
\newtheorem{itassumptionr}{Assumption (rev.)}
\newtheorem{itremark}{Remark}
\newtheorem{itproposition}{Proposition}
\newtheorem{itresult}{Result}
\newtheorem{itproblem}{Problem}

\newcommand{\subscr}[2]{#1_{\textup{#2}}}

\newcommand{\eps}{\varepsilon}

\newcommand{\be}{\begin{equation}}
\newcommand{\ee}{\end{equation}}
\newcommand{\ba}{\begin{array}}
\newcommand{\ea}{\end{array}}

\newcommand{\E}{\mathcal{E}}	

\newcommand{\1}{{\bf 1}}

\newcommand{\e}{{\rm e}}

\ifCLASSOPTIONcompsoc
  \usepackage[nocompress]{cite}
\else
  \usepackage{cite}
\fi
   \graphicspath{{pics/}{pics_new_paolo_geert/}}
\usepackage{tikz}
\usetikzlibrary{positioning, arrows,shapes,trees,calc,through}

\usepackage{amsmath}
\begin{document}
\title{Stochastic String Stability of Vehicle Platoons via Cooperative Adaptive Cruise Control with Lossy Communication
\thanks{
The authors would like to thank Elham Semsar-Kazerooni for insightful discussions on the topics of this paper.
}}
\author{Francesco Acciani,%
\IEEEcompsocitemizethanks{F.\ Acciani, P.\ Frasca, G.\ Heijenk  and A.\ Stoorvogel are with the Faculty of Electrical Engineering, Mathematics and Computer Science, University of Twente, Enschede, Netherlands.} 
Paolo Frasca, 
\IEEEcompsocitemizethanks{ P.\ Frasca is with Univ.\ Grenoble Alpes, CNRS, Inria, Grenoble INP, GIPSA-lab, Grenoble, France. 
E-mails: paolo.frasca@gipsa-lab.fr.}%
Geert Heijenk, 
Anton Stoorvogel
}


\maketitle

\begin{abstract}
This paper is about obtaining stable vehicle platooning by using Cooperative Adaptive Cruise Control when the communication is unreliable and suffers from message losses. We model communication losses as independent random events and we propose an original design for the cooperative controller, which mitigates the effect of losses. This objective is obtained by a switching controller that has a twofold objective: on the one hand, it promotes both plant stability and string stability of the average error dynamics by an $\mathcal{H}_\infty$ approach, and on the other hand it minimizes the variance around the average. We show by simulations that the proposed controller is able to compensate even for high probability of losses. 
\end{abstract}

\IEEEdisplaynontitleabstractindextext

%
\IEEEpeerreviewmaketitle

\ifCLASSOPTIONcompsoc
\IEEEraisesectionheading{\section{Introduction}\label{sec:introduction}}
\else

\section{Introduction}
\label{sec:introduction}
\fi

A platoon is a group of vehicles that move close together at the same speed: automated vehicle platooning has been heralded since the eighties as an enabler of effective road usage. 
Crucial to the effectiveness of platoons is their ability to remain {\em string stable}, that is, to dampen disturbances that may affect the motion of vehicles.

Adaptive Cruise Control systems allow vehicles to adapt their dynamics to their surroundings by taking measurements of relative distances and speeds. ACC systems, which are now becoming widespread in the market, can be used to create platoons that however may not be string stable in many situations.
Cooperative Adaptive Cruise Control provides an effective way to stabilize platoons, with performance that is superior to ACC without cooperation: besides inter-vehicle relative measurements of distance and speed, cooperation requires that vehicles communicate to exchange additional internal variables, such as their current control inputs.  

Cooperation in CACC thus hinges on the ability of the vehicles to communicate effectively and consequently is prone to disruption if communication is problematic, for instance if it is severely affected by delays or packet losses.
These vulnerabilities have been observed by several researchers, who have also proposed both robustness analyses and designs to mitigate the disruptions~\cite{Liu:2001aa,Lei:2011aa,Ploeg:2015,Oncu:2014aa,Qin:2017aa}.

\paragraph{Contribution}
Our work concentrates on the effect of message losses, modeled as stochastic independent events, and on the design of a cooperative controller that mitigates their effect.
We frame the issue of robustness to packet losses as the stability of a string of stochastic systems. 
Based on this reflection, we design a distributed controller to guarantee that the average behavior of the platoon is string stable and that the variance around the average is minimized. This twofold requirement is meant to make disturbances be dampened along most of the realized trajectories of the stochastic system. This design is achieved by means of $\mathcal{H}_\infty$ control methods in the state space that are applied on the expected dynamics. Despite the large literature on communication and CACC, our work seems the first to incorporate the stochastic nature of losses in the design of a string-stabilizing controller that has the purpose of mitigating their effects. 
The obtained distributed controller is tested in simulation showing its ability to compensate for high probability of losses while keeping the platoon tightly together.

\paragraph{Literature review} 

The importance of communication for string stabilization is supported by evidence from work spanning across at least two decades~\cite{Liu:2001aa,Arem:2006aa,Lei:2011aa,Pates:2017aa}. When CACC is deployed by unreliable communication, researchers have observed that ``as the probability of data loss increases, the behavior of the platoon becomes unacceptable [...], unless the time headway constant is also increased''~\cite{Vargas:2018aa}. 

In this paper we make well-accepted assumptions about the communication between the vehicles. Messages are sent at uniformly-spaced sample times and failure (loss) events are stochastic, in accordance with most literature on communication networks. 
On this matter, it is worth noting that even though string stability notions have been studied for half a century~\cite{Peppard:1974,Pant:2002aa,Swaroop:1996aa}, relatively little attention has been devoted to stochasticity in this context. 
Deterministic models are more popular in the control community~\cite{Harfouch:2018aa,Molnar:2018,merco:2019}, but \cite{Qin:2017aa} has performed string stability analysis for mean and covariance dynamics and \cite{Qin:2015} has proposed stochastic string stability notions for platoons with random delays. Recently, event-driven communication has also been proposed in this context~\cite{Dolk:2017aa}.  

Regarding our design tools, we would like to mention that $\mathcal{H}_\infty$ design is also used in~\cite{Ploeg:2014bb,Zheng:2018aa} for platooning applications and in \cite{Li:2010ab,Li:2018aa} for general stochastic networked control systems. Another $\mathcal{H}_\infty$ design, which includes packet losses, can be found in~\cite{Seiler:2005}, while early designs for mean square stability with packet losses can be found in~\cite{Seiler:2001}.

Finally, in a recent paper \cite{Acciani:2018ab}, we approached the issue of string stability on lossy networks by a frequency-domain design, using a dynamic controller together with an unknown input observer that reconstructs the missing inputs. Related references include \cite{Naus:2010} on frequency-domain designs and \cite{Molnar:2018} on using predictors.

\paragraph{Paper structure} The rest of the paper is organized as follows. In Section~\ref{sect:models} we describe the full dynamical model of the platoon. In Section~\ref{sect:string} we discuss the suitable notion of string stability that we set as design objective. In Section~\ref{sect:control-design} we describe our control architecture with all its components. In Section~\ref{sect:simulations} we test the performance of our controller in simulation, showing its effectiveness. We conclude with some reflections and perspectives in Section~\ref{sect:outro}.

\section{Vehicle and platoon models}\label{sect:models}
This section describes our model for the platoon, including both the motion model of the vehicles, the error dynamics with respect to the control objective, and the assumptions about the available measurements and communications. 

\subsection{Vehicle model and spacing policy}
We take a well-established~\cite{Sheikholeslam:1990aa} dynamical model of the single vehicle, which takes into account only its longitudinal motion. We let scalar $q_i(t)$ be the position of vehicle $i$, $v_i(t)$ its speed and $a_i(t)$ its acceleration and we assume the dynamics	
	\begin{align}
		\dot{a}_i(t) = -\frac{1}{\tau}a_i(t) + \frac{1}{\tau}u_i(t-\phi),  \label{eq:dynamics}
	\end{align}
where $\tau>0$ is the vehicle time constant, and $\phi>0$ the internal delay of the system. As usual, $u_i(t)$ is the input for our system. Since the constants do not depend on $i$, the platoon is homogeneous: we mention {\it en passant} that heterogeneity in platoons is s relevant issue and an active topic of research~\cite{Harfouch:2018aa}.

The objective of the controller is to maintain the distance between vehicles at a certain reference: we adopt here a speed-dependent time spacing policy, according to~\cite{Swaroop:1994aa}. 
The desired distance between vehicle $i$ and $i-1$ is then defined as
	\begin{align*}
	d_{r,i}(t) = R_i + h v_i(t)
	\end{align*}
	where the positive constant $h$ is a time-headway (lower values for $h$ represent shorter distances between vehicles) and $R_i$ is the standstill desired distance. In the following of this paper $R_i$ will be neglected, without loss of generality: it is always possible to find a coordinate transformation equivalent to the choice $R_i = 0$. 
%
We can now define the error for vehicle $i$ as
	\begin{align}
	e_i(t) = d_i(t) - d_{r,i}(t) = q_{i-1}(t)-q_i(t) - hv_i(t)\label{eq:error} 
	\end{align}
	where $d_i(t) = q_{i-1}(t)-q_i(t)$ is the relative distance between vehicles $i$ and $i-1$. 
Finally, we observe that the vector state
	\begin{align*}
	x(t) = \left[\begin{array}{c}e_i(t)\\ \dot{e}_i(t)\\ \ddot{e}_i(t) + \frac{h}{\tau}u_i(t-\phi)\end{array}\right]
	\end{align*}
follows the closed dynamics
	\begin{align}\label{eq:error-dynamics}
	\dot{x}(t) =  
	\left[ \begin{matrix}
	0 	&1 		&	0\\
	0 	&0		& 	1\\
	0	&0		&-\frac{1}{\tau}\\
	\end{matrix}\right]x(t) +  
	\left[\begin{matrix}
	0\\
	-\frac{h}{\tau}\\
	\frac{h-\tau}{\tau^2}\end{matrix}\right]u_i(t-\phi) +
	\left[\begin{matrix}
	0\\
	0\\
	\frac{1}{\tau}\end{matrix}\right]u_{i-1}(t-\phi),		
	\end{align}
	which is consistent with~\eqref{eq:dynamics}-\eqref{eq:error} and which we shall refer to as the {\em error dynamics}.

\subsection{Measurements and communication}
We assume that every vehicle  is able to measure the relative distance and relative speed with respect to the preceding vehicle (say, by a radar sensor), as well as its own absolute speed and acceleration. These measurements permit to reconstruct the first and second components of the state, i.e. $e_i(t) = d_i(t) - hv_i(t)$ and $\dot e_i(t) = v_{i-1}(t) - v_i(t) - ha_i(t)$. Instead, there is no simple way to measure relative acceleration or jerk: therefore we can not measure the third component of the state. 
Therefore, we define the \emph{measured} output $y(t)$ as:
	\begin{align*}
	y_i(t) = Cx_i(t-\psi)+w(t)
	\end{align*}
	where $C = \left[\begin{matrix}1 &0  &0\\ 0 &1 &0\end{matrix}\right]$, positive scalar $\psi$ is a measurement delay 
	and $w(t)$ is measurement noise.

Since we are looking for a cooperative controller, we assume that the vehicles are also able to communicate. 
We thus define a \emph{communicated} output, which is the input of the previous vehicle: this is a noise-free quantity which however suffers from losses and transmission delay:
	\begin{align*}
	{\subscr{y}{comm}}_i(t) = f(u_{i-1}(t-\theta))
	\end{align*}
	where $f\left(u_{i-1}(t-\theta)\right)$ models the network unreliability:
	\begin{align*}
	f\left(u_{i-1}(t-\theta)\right) = \left\{\begin{array}{cl}
	u_{i-1}(t-\theta) &\quad\text{if communication available}\\
	0			  &\quad\text{otherwise}	\end{array}\right.
	\end{align*}
%
A scheme representing which quantities are used by the vehicles can be found in Fig.~\ref{fig:connection_car}: typical values for the parameters of the continuous time systems are \cite{Ploeg:2011aa}:
	\begin{align*}
	\tau = 0.1s, \quad \phi = 0.2s, \quad \theta = 0.02s, \quad \psi = 0.05s.
	\end{align*}

	\begin{figure*}
    \resizebox{\textwidth}{!}{
		\begin{tikzpicture}
		\tikzstyle{box} = [rectangle, draw = black, minimum width = 1cm, minimum height = 1cm, node distance = 1.4cm, rounded corners]
		
		\node [box] (enginei)[right] at (5,2) {Engine};
		\node [box] (controlleri)[right] at (2,2) {Controller};
		\node (cari) at (0,0)[opacity=.2, above right] {\reflectbox{\includegraphics[height=3.8cm]{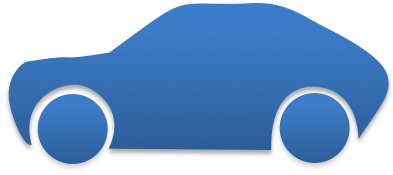}}};
		\node (qi) [above left = .25 cm of controlleri] {\small$q_{i-1}-q_i$};
		\node (vi) [left = .25 cm of controlleri] {\small$v_{i-1}-v_i$};
		\node (ai) [below left = .25 cm of controlleri] {\small$a_i$};
		\node (ui) [right = .5cm of controlleri, above] {\small$u_i$};
		\draw [->] (qi) -- (controlleri.north west);
		\draw [->] (vi) -- (controlleri.west);
		\draw [->] (ai) -- (controlleri.south west);
		\draw [->] (controlleri.east) -- (enginei.west);
		
		\node [box] (engine)[right] at (15,2) {Engine};
		\node [box] (controller)[right] at (12,2) {Controller};
		\node (car) at (10,0)[opacity=.2, above right] {\reflectbox{\includegraphics[height=3.8cm]{pics/car.png}}};
		\node (q) [above left = .25 cm of controller] {\small$q_{i-2}-q_{i-1}$};
		\node (v) [left = .25 cm of controller] {\small$v_{i-2}-v_{i-1}$};
		\node (a) [below left = .25 cm of controller] {\small$a_{i-1}$};
		\node (u) [right = .5cm of controller, above] {\small$u_{i-1}$};
		\draw [->] (q) -- (controller.north west);
		\draw [->] (v) -- (controller.west);
		\draw [->] (a) -- (controller.south west);
		\draw [->] (controller.east) -- (engine.west);
		
		\node (right dummy) [above of= u] {};
		\node (top dummy) [above of = controlleri] {};
		\node (left dummy) [above of = top dummy] {};
		\draw [->](u) -- (right dummy.south) -- (top dummy.north) -- (controlleri.north);
		
		\node (left dummy 2) [above = 2cm of ui]{};
		\node (left left dummy) [left =3 cm of left dummy 2]{};
		\node (right dummy 2) [above = 2cm of controller]{};
		\node (right right dummy) [right =3 cm of right dummy 2]{};
		\draw [dashed, ->] (ui) -- (left dummy 2.south) -- (left left dummy.south); 
		\draw [dashed,<-] (controller) -- (right dummy 2.south) -- (right right dummy.south);
		\node (ui1) [above of =controller,below right, opacity = .5] {$u_{i-2}$};
		\end{tikzpicture}
	}
		\caption{Diagram of the platoon -- detail of the interconnection between two vehicles.\label{fig:connection_car}}
	\end{figure*}

\subsection{Communication model and discrete-time dynamics}
After modelling the continuous-time dynamics, we are going to convert it to a discrete-time one by choosing a suitable inter-sampling time $T_s$.  Indeed, the discrete-time framework is more suitable to accommodate for the communication between vehicles, which is inherently a discrete time phenomenon, including the stochastic losses.
To model the communication, we assume that transmissions occur at the sampling times and that, upon each transmission, each vehicle is able to receive the information that is sent by the preceding one with probability $1-p$ (thus $p$ is the probability for each sent message to be lost). Transmissions are assumed to be time and space independent, therefore a loss happening at some time for one vehicle does not influence future or current losses for any other vehicle in the platoon.
Furthermore, discrete time allows us to easily incorporate delays in our analysis, so long as we assume that delays are multiples of the sampling time.

Based on the above considerations, we choose a inter-sampling time $T_s=0.01s$ and 
 $d = \frac{\phi}{T_s}$, $m = \frac{\psi}{T_s}$, and $r = \frac{\theta}{T_s}$.
We thus obtain the model 
	\begin{align}
	x_i(k+1) = Ax_i(k) + Bu_i(k-d) +Eu_{i-1}(k-d) \label{eq:ss_general}
	\end{align}
	where
	\begin{align*}
	A &= \left[ \begin{matrix}
	1 	&T_s 	&	\tau T_s - \tau^2\left(1-e^\frac{-T_s}{\tau} \right)\\
	0 	& 1 		& 	\tau\left(1-e^\frac{-T_s}{\tau} \right)\\
	0	&0		&e	^\frac{-T_s}{\tau}
	\end{matrix}\right] \\
	B &= \left[\begin{matrix}
	-\frac{T_s^2}{2} + T_s(\tau-h)-\tau(\tau-h)\left(1-e^{-\frac{T_s}{\tau}}\right)\\
	-T_s + (\tau-h)\left(1-e^{-\frac{-Ts}{\tau}}\right)\\
	-\frac{\tau-h}{\tau}\left(1-e^{-\frac{T_s}{\tau}}\right)
	\end{matrix}\right]\\
	E &= \left[\begin{matrix}
	\frac{T_s^2}{2} - T_s\tau + \tau^2\left(1-e^\frac{-T_s}{\tau} \right) \\
	T_s - \tau\left(1-e^\frac{-T_s}{\tau} \right)		\\
	1-e^\frac{-T_s}{\tau}	\\
	\end{matrix}\right]
	\end{align*}
and by a slight notational abuse the discrete-time variables are denoted by the same letters as the corresponding continuous-time variables.
	The 
	full model then reads
	\begin{align}\label{eq:ss_complete}
	x_i(k+1) 	&= Ax_i(k) + Bu_i(k-d) +Eu_{i-1}(k-d) 
	\nonumber\\
	y_i(k) 	&= Cx_i(k-m) 
	\\
	{\subscr{y}{comm}}_i(k) 	&=	u_{i-1}(k-r) \, \delta_i(k)\nonumber
	\end{align}
	where $\delta_i(k)$ is a Bernoulli random variable with mean $1-p$.
	For controller design purposes, we shall disregard 
	measurement delays (which can be incorporated in the larger input delay) and transmission delays (which we assume to be small). We thus effectively work on 	
	\begin{align}\label{eq:ss_for-design}
	x_i(k+1) 	&= Ax_i(k) + Bu_i(k-d) +Eu_{i-1}(k-d)\nonumber\\
	y_i(k) 	&= Cx_i(k) \\
	{\subscr{y}{comm}}_i(k) 	&=	u_{i-1}(k) \, \delta_i(k)\nonumber
	\end{align}   The general model~\eqref{eq:ss_complete}, however, shall be used to test the controller in simulation.

\section{String stability and control objectives}\label{sect:string}
In order to make every vehicle in the platoon smoothly follow the preceding one, our control objective is twofold: stabilizing the error~\eqref{eq:error} to zero and ensuring a {\em string stability} property. This latter notion refers to the uniform boundedness of the states or, equivalently, to the dampening of disturbances along the string of vehicles. Indeed, disturbance amplification is not only a risk for a safe operation of the platoon, but it also compromises traffic flow stability and throughput \cite{Oncu:2014ab}. If the metric to measure disturbances is $\mathcal{L}_2$, then the resulting notion is the so-called $\mathcal{L}_2$-string stability that requires 
\begin{align}
\frac{\Vert u_i\Vert_{\mathcal{L}_2}}{\Vert u_{i-1}\Vert_{\mathcal{L}_2}}<1. \label{eq:L2norm}
\end{align}
For a proper definition of a feasible and useful control objective, we make three choices that correspond to three crucial points to account for. Firstly, we need to ensure {\em both plant stability} (that is, stability of~\eqref{eq:error-dynamics}) for each vehicle {\em and string stability} for the whole platoon. To this purpose, 
we consider a combination of input and error as $$ z_i(k) = [\varepsilon e_i(k), r\, u_i(k)]^\top\!\!$$ (with suitable positive constants $\varepsilon, r$) and replace \eqref{eq:L2norm} by
\begin{align}
\frac{\Vert z_i\Vert_{\mathcal{L}_2}}{\Vert u_{i-1}\Vert_{\mathcal{L}_2}}<1, \label{eq:z-stability}
\end{align}
which --contrarily to \eqref{eq:L2norm}-- promotes both plant and string stability. 
Secondly, instead of imposing the ratio in~\eqref{eq:z-stability} to be smaller than one, we aim to make it as small as possible: 
\begin{align}
\min \frac{\Vert z_i\Vert_{\mathcal{L}_2}}{\Vert u_{i-1}\Vert_{\mathcal{L}_2}}. \label{eq:H-inf}
\end{align}
This formulation can take advantage of the solid literature on $\mathcal{H}_\infty$ control design devoted to minimizing this cost.

Thirdly and most importantly, our system is {\em stochastic} in nature and therefore it requires us to adapt the deterministic definition of string stability that we have given above.
Ideally, one would like that disturbances be attenuated along every or almost every trajectories. However, such a requirement would be too restrictive for performance in terms of the necessary headway $h$. On the contrary, requiring disturbance attenuation in expectation only would allow for a small headway $h$ but would actually be too weak a requirement, since trajectories be free to significantly deviate from the average.
In view of this robustness/performance trade-off in the definition of string stability in the stochastic setting, in our work we aim at ensuring an acceptable behavior of the stochastic string with good performance by a two-step approach: we impose string stability of the {\em expected trajectory} and we {\em minimize the variance} of the trajectories around their expectation.

\section{Control design}\label{sect:control-design}

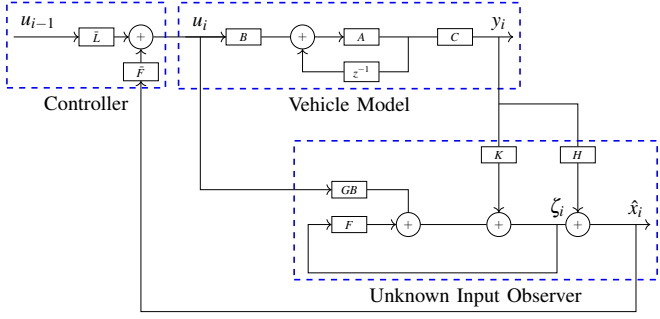
\begin{figure}
\resizebox{\columnwidth}{!}{
\begin{tikzpicture}
	\tikzstyle{tf} = [rectangle, draw = black, minimum width= 1cm, minimum height = 1/1.618 cm, scale = .6]
	\tikzstyle{sum} = [draw=black, shape=circle, minimum width = .2cm, minimum height = .2cm,scale = 0.6]
	\tikzstyle{signal} = [above]
	
	\node (A) [tf] at (0,0) {$A$};
	\node (z) [tf, below of=A] {$z^{-1}$};
	\node (dummyx) [right =0.5cm of A, above] {};
	\node (C) [tf, right =1cm of A] {$C$};
	\node (sum) [sum, left=0.5 cm of A]{$+$};
	\node (B) [tf, left = 0.5cm of sum] {$B$};
	\node (u) [signal, left of = B, above right] {$u_i$};
	\node (y) [signal, right of = C, above left] {$y_i$};
	
	\draw [->] (dummyx.south) |- (z) -| (sum.south);
	\draw [->] (A)--(C)--(y.south east);
	\draw [->] (u.south west) -- (B) -- (sum) -- (A);

	\node (rettangolo) [draw,blue,thick, dashed, fit=(u) (z) (u) (y),label=below:Vehicle Model]{};
	
	\node (dummy) [below =1 cm of rettangolo] {};

	\node (GB) [tf,below of = dummy] {$GB$};
	\node (F) [tf,below of = GB] {$F$};
	\node (sum2) [sum, right = 0.5cm of  F] {$+$};
	\node (K) [tf] at (dummy-|y){$K$};
	\node (H) [tf] [right = 0.75cm of K] {$H$};
	\node (topH) [above of = H] {};
	\node (boty) at (y|-topH){};
	\node (sum3) [sum] at (F-|y) {$+$};
	\node (zeta) [signal, right of = sum3, above] {$\zeta_i$};
	\node (sum4) [sum] at (sum3-|H){$+$};
	\node (xhat) [signal, right of=sum4, above] {$\hat{x}_i$};
	\node (dummyF) [below left of = F]{};
	\draw [->] (y.south) --(K)-- (sum3);
	\draw [->] (boty.south) -| (H) -- (sum4);
	\draw [->] (GB)-|(sum2)--(sum3)--(sum4)--(xhat.south east);
	\draw [->] (F)--(sum2);
	\draw [->] (u) |- (GB);
	\draw [->] (zeta) |- (dummyF.south) |- (F);
	
	\node (rettangolo2) [draw,blue,thick, dashed, fit=(dummy) (dummyF) (dummyF) (xhat),label=below:Unknown Input Observer]{};

	\node (sum5) [sum, left =1.25 cm of B] {$+$};
	\node (F2) [tf, left =0.25 of sum5] {$\bar L$};
	\node (F1) [tf, below  of=  sum5]{$\bar F$};
	\node (uprec) [signal, left of =F2, above] {$u_{i-1}$};
	
	\draw [->] (F1.north)-|(sum5);
	\draw [->] (uprec.south west) -- (F2);
	\draw [->] (F2) -- (sum5);
	\draw [->] (sum5) --(B);
	
	\node (dummyxhat) [below =1.25 cm of xhat] {};
	\node (dummyfeedback) at (F1|-dummyxhat){};;
	\draw [->] (xhat.south) -- (dummyxhat.south) -- (dummyfeedback.south) -- (F1);
	
	\node (rettangolo2) [draw,blue,thick, dashed, fit=(uprec) (F1) (uprec) (F1),label=below:Controller]{};
\end{tikzpicture}
}
\caption{Control architecture with $\mathcal{H}_\infty$ controller and state observer, that is, $u_i(k)=\bar F \hat x_i(k) + \bar L u_{i-1}(k)$). This controller is applied to the average dynamics. The full control architecture includes the switching logic and the lifting block.}\label{fig:control-structure}	
\end{figure}

The design of our cooperative controller is done through several phases, which roughly correspond to the main components of the control architecture that is summarized in Fig.~\ref{fig:control-structure}.
The vehicle is fed by $u_i(k)$, computed by the controller. The controller provides $u_i(k)$ by using the input of the preceding vehicle $u_{i-1}(k)$ when available, the estimate of the state from an unknown input observer, and the delayed samples of the inputs $u_i(k)$ and $u_{i-1}(k)$. More specifically, to model the effects of delays in the vehicle, we study a lifted system, which increases the size of the state. However, it is possible to design the unknown input observer based on the small, delay-free system, and enrich the output of the unknown input observer with previous measurements of the inputs $u_i(k)$ and $u_{i-1}(k)$ when available, to fully reconstruct the lifted state, which is controller's input.
To keep the synthesis and the design of the controller simple, we present all the ingredients for the controller synthesis without taking delays into account, while the actual computation of the controller takes into account the lifted, high-dimensional, system.
We shall first design a cooperative controller for the generic vehicle $i$ in the \emph{full information} case, that is, assuming perfect knowledge of $x_i(k)$ (state-feedback) and $u_{i-1}(k)$ (ideal communication without losses). This preliminary controller, in the form
	\begin{align}\label{eq:non-switching_control_law}
u_i(k)=\bar F x_i(k) + \bar L u_{i-1}(k),\end{align}  has the objective of ensuring both closed-loop stability of the error dynamics and string stability of the platoon. Next, we take into account the communication limitations and use $u_{i-1}(k)$ only when it is available; the control law therefore becomes the 
following switching one:
	\begin{align}\label{eq:control_law}
	u_i(k) = \left\{
	\begin{array}{ll}
	F_1{x}_i(k) + L u_{i-1}(k) & \quad\text{if } u_{i-1}(k) \text{ available} \\
	F_2{x}_i(k)					& \quad\text{otherwise}
	\end{array}
	\right.
	\end{align}
We adapt the controller to the stochastic nature of our system in two steps: (i) we study the average or {\em expected} dynamics of the system and design the controller in such as way that the expected behavior matches the ideal behavior (without losses); (ii) we include the criterion of minimizing the {\em variance} to account for the dispersion of the actual trajectories.

Finally, we design an {\em unknown input observer} to produce an estimate of the state $\hat{x}_i(k)$ that can be used for state-feedback and we incorporate delays by 
applying the aforementioned steps to 
a suitably {\em lifted} system. All these steps are suitably detailed in the following subsections.

\subsection{Lossless communication and state-feedback}\label{sect:full-info-controller}
To design the controller, we study the single vehicle dynamics in the lossless case. From now on, we drop the subscripts and we use $\xi(k)$ to denote $u_i(k)$, the input of the $i^{th}$ vehicle, and $\nu(k)$ to denote $u_{i-1}(k)$, the input of the previous vehicle $i-1$.
	The model of the single vehicle that we use to design a controller is the following:
	\begin{align}\label{eq:system_control}
	\left\{
	\begin{array}{rl}
		x(k+1) 	&= A\,x(k) + B\xi(k) + E\nu(k) \\
		y(k)	&= C\, x(k) \\
		z(k) 	&= C_zx(k) + D_\xi \xi(k)
	\end{array}\right.
	\end{align}
The variable $y(k)$ is the measured output defined by \begin{align*}
		C = \left[\begin{array}{ccc}1 &0&0 \\ 0&1&0 \end{array}\right],
	\end{align*}
	while $z(k)$ is a \emph{performance} output: it is not measured, but it is used to design the controller. 
	The matrices $C_z$ and $D_\xi$ are:
	\begin{align*}
		C_z = \left[\begin{array}{ccc}\varepsilon &0 &0 \\ 0&0&0 \end{array}\right], \quad D_\xi = \left[\begin{array}{c}0\\r\end{array}\right]
\end{align*}
	so that the performance output is a combination of error and local input: 
	\begin{align*}
	z(k) = [\varepsilon e(k), r\,\xi(k)]^\top\!\!.
	\end{align*}
In order to promote both asymptotic stability and string stability, we define the $\mathcal{H}_\infty$ control objective of minimizing the norm
	\begin{align}
	\label{eq:zinequality}
	\Vert H \Vert_{\mathcal{H}_\infty}  = \frac{\Vert z\Vert_{\mathcal{L}_2}}{\Vert \nu \Vert_{\mathcal{L}_2}}.
	\end{align}
	Comparing \eqref{eq:zinequality} with \eqref{eq:L2norm}, the choice of $z$ becomes clear: we put a small weight on $e_i(k)$ so that $z(k)$ is close to $\xi(k)$.
	In order to ensure string stability, we hope that this norm can be made smaller than one by the $\mathcal{H}_\infty$ design.	
To solve this design problem, we make the assumption that the full state is available for feedback, thereby focusing on system 
	\begin{align}\label{eq:system_control_state-feedback}
	\left\{
	\begin{array}{rl}
		x(k+1) 	&= Ax(k) + B\xi(k) + E\nu(k) \\
		y(k)	&= x(k) \\
		z(k) 	&= C_zx(k) + D_\xi \xi(k)
	\end{array}\right.
	\end{align}
On this system, the $\mathcal{H}_\infty$ problem can be solved by applying the following result from~\cite{Stoorvogel:1992aa}.
\begin{lemma}[$\mathcal{H}_\infty$ design]\label{lem:Hinf-anton}
	Consider the following system
	\begin{align}\label{eq:system_h}
	\left\{
	\begin{array}{rl}
	x(k+1) 	&= Ax(k) + B\xi(k) + E\nu(k)\\
	z(k) 	&= Cx(k) + D_1\xi(k) 
	\end{array}\right.
	\end{align}
	and assume that system $(A,B,C,D_1)$ has no invariant zeroes on the unit circle. 
	Then, a controller in the form $$\xi(k) = \overline{F} x(k) + \bar L \nu(k)$$ exists, such that the closed-loop transfer function from $\nu$ to $z$ has $\mathcal{H}_\infty$-norm less than one, if and only if a symmetric matrix $P\geq 0$ exists, such that:
	\begin{enumerate}
		\item $V = D_1^\top\!\!D_1 + B^\top\!\!PB > 0$ 
		\item $R = I  - E^\top\!\!PE + E^\top\!\!PB \,V^{-1}\, B^\top\!\!PE > 0$
		\item $P$ satisfies the following Riccati equation: 
		$$P = A^\top\!\!PA+C^\top\!\!C-\begin{bmatrix}B^\top\!\!PA+D_1^\top\!\!C\\E^\top\!\!PA\end{bmatrix}^\top\!\! \!\! G(P)^{-1}\begin{bmatrix}B^\top\!\!PA+D_1^\top\!\!C\\E^\top\!\!PA\end{bmatrix},$$
		where $G(P) = \begin{bmatrix}D_1^\top\!\!D_1 & 0 \\ 0& - I\end{bmatrix} + \begin{bmatrix}B^\top\!\!\\E^\top\!\!\end{bmatrix}P\begin{bmatrix}B&E\end{bmatrix}$
	\end{enumerate}
	If such $P$ matrix exists, then the static feedback matrices $\overline{F}$ and $\overline{L}$ can be chosen as 
	\begin{align*}
	\overline{F} &= -[D_1^\top\!\!D_1+B^\top\!\!PB]^{-1} [B^\top\!\!PA+D_1^\top\!\!C]\\
	\overline{L} &= -[D_1^\top\!\!D_1+B^\top\!\!PB]^{-1} B^\top\!\!PE.
	\end{align*}
\end{lemma}
A general procedure to compute matrix $P$ can be found in \cite{Chen:1994aa}.
%
In the next subsections we shall see more precisely how this deterministic design problem can be useful in our stochastic lossy system~\eqref{eq:ss_for-design}.

\subsection{Controlling the expectation}
	By using the model \eqref{eq:system_control} for the single vehicle dynamics, and assuming that the communication between vehicles is modelled by a Bernoulli process,  our system consists of $n$ sub-systems, each of which can be described by the dynamics~\eqref{eq:ss_for-design}:
	\begin{align}\label{eq:lossy_ss}
	\left\{\begin{array}{rll}
	x(k+1) &= Ax(k) + B\xi(k) + E\nu(k)\\
	y(k) & = C x(k) \\
	\subscr{y}{comm}&=\delta(k) \nu(k) 
	\end{array}\right.
	\end{align}
	where $\delta(k) = 0$ when the communication from the preceding vehicle is lost, and $\mathds{P}\left[\delta(k) = 1\right] = 1-p$. 
	We recall that the losses, i.e.\ $\delta(k)$, are time and space independent: the losses experienced by different vehicles are uncorrelated, as well as subsequent losses experienced by a single vehicle.

	The following simple result gives us the average behavior of the lossy system~\eqref{eq:lossy_ss} interconnected with a switching controller. 
\begin{lemma}[Expected dynamics]\label{lem:Ex}
		The expected value $\mathds{E}[x(k)]$ of the state of system~\eqref{eq:lossy_ss}
		interconnected with a stochastic switching controller in the form:		
		\begin{equation}\label{eq:general_switching_controller}
		\xi(k) =\delta(k)\big(F_1x(k)+L\subscr{y}{comm}(k)\big)		+\big(1-\delta(k)\big) F_2\, x(k)
		\end{equation}
		has the same dynamics as the state of system~\eqref{eq:lossy_ss} when interconnected with a non-switching deterministic controller in the form:
		\begin{align*}
		\xi(k) = \overline{F}x(k) + \overline{L}\nu(k) 
		\end{align*}
		where 
		\begin{align*}
		\overline{F} 	&= (1-p)F_1 + pF_2\\
		\overline{L} 	&= (1-p)L
		\end{align*}
	\end{lemma}

	This lemma can be used to find a switching controller such that the expected dynamics of the lossy system is the same as a given deterministic system. 
In our design, we shall impose that {\em the expected system coincides with the nominal system without losses~\eqref{eq:system_control_state-feedback}}. In order to exploit the state-feedback controller provided by Lemma~\ref{lem:Hinf-anton}, we compute $\overline{L}$ and $\overline{F}$ by applying Lemma~\ref{lem:Hinf-anton} with $C=C_z$, $D_1=D_\xi$. Matrix $F_1$ will be chosen later in such a way to minimize covariance, as detailed in the next subsection.

%

\subsection{Covariance minimization}
	Thanks to Lemma \ref{lem:Ex}, we are able to guarantee the expected behaviour of the platoon to be string stable, provided that it is possible to find for the lossless system~\eqref{eq:system_control_state-feedback} a state feedback $\xi=\overline{F}x+\overline{L}\nu$ such that the $\mathcal{H}_\infty$ gain is less than one.
Now, we focus on minimizing the covariance of the error, in order to keep the trajectories to be  close to the expected, string stable, behavior.	
Therefore we want to minimise the variance
	\begin{align}\label{eq:covariance_min}
	\mathds{E}\left[\Vert \left(x-\mathds{E}[x]\right) \Vert_2^2\right] 
	\end{align}
\begin{theorem}[Controller design]
Consider the dynamics 
$$ 	x(k+1) = Ax(k) + B\xi(k) + E\nu(k)$$  with the stochastic control law 
	\begin{align*}
	\xi(k)=\left\{
	\begin{array}{ll}
	F_1x(k) + L\nu(k) &\text{with probability $1-p$}\\
	F_2x(k)				&\text{with probability $p$}\\	
	\end{array}
	\right.
	\end{align*}
and the same dynamics with the nominal deterministic control law $\xi(k)=\bar F x(k) + \bar L \nu(k) $.
Assume that \begin{align}\label{eq:F1}
	F_1 =& \overline{F} - pL\mathds{E}[\nu]\mathds{E}[x]^\top\!\!\left(\mathds{E}[\tilde{x}\tilde{x}^\top\!\!]+\mathds{E}[x]\mathds{E}[x]^\top\!\!\right)^{-1}\\
	F_2 =& \frac{1}{p}\left(\overline{F}-(1-p)F_1\right) \label{eq:F2}\\
	L = &\frac{1}{1-p}\overline{L},\label{eq:L}
	\end{align}
	where $\tilde{x}(k)=x(k)-\mathds{E}\left[x(k)\right]$.
Then, the expectation of the stochastic dynamics follows the nominal dynamics and cost~\eqref{eq:covariance_min} is minimized. \end{theorem}
	
\begin{proof}


We can compute the dynamics of the expected state
	\begin{align}
	\nonumber
	\mathds{E}[x(k+1)] =& A\mathds{E}[x(k)]+B\left((1-p)F_1+pF_2\right)\mathds{E}[x(k)]\\&+(1-p)BL\mathds{E}[\nu(k)] \label{eq:exalpha}
	\end{align}
and, by some lengthy manipulations using \eqref{eq:F2} that are not reported, the dynamics of its covariance
	\begin{align*}
\mathds{E}\left[\tilde{x}_+\tilde{x}^\top\!\!_+\right]=&A\mathds{E}[\tilde{x}\tilde{x}^\top\!\!]A^\top\!\! + E\mathds{E}[\tilde{\nu}\tilde\nu^\top\!\!]E^\top\!\! + 
	2(1-p)BL\mathds{E}[\tilde\nu\tilde\nu^\top\!\!]E^\top\!\!\\&+ (1-p)BL\mathds{E}[\tilde\nu\tilde\nu^\top\!\!]L^\top\!\!B^\top\!\! +
	 p(1-p)BL\mathds{E}[\nu]\mathds{E}[\nu]^\top\!\!L^\top\!\!B^\top\!\!  \\ &-B\overline{F}\mathds{E}[x]\mathds{E}[x]^\top\!\!\overline{F}^\top\!\!B^\top\!\! + 2B\overline{F}\mathds{E}[\tilde{x}\tilde{x}^\top\!\!]A^\top\!\! \\
	&- 2(1-p)B\overline{F}\mathds{E}[x]\mathds{E}[\nu]^\top\!\!L^\top\!\!B^\top\!\! \\& +\frac{1}{p}B\overline{F}\left(\mathds{E}[\tilde{x}\tilde{x}^\top\!\!]+\mathds{E}[x]\mathds{E}[x]^\top\!\!\right)\overline{F}^\top\!\!B^\top\!\! \\
	&+\frac{1-p}{p}BF_1\left(\mathds{E}[\tilde{x}\tilde{x}^\top\!\!]+\mathds{E}[x]\mathds{E}[x]^\top\!\!\right)F_1^\top\!\!B^\top\!\!
	\\&+2(1-p)BF_1\mathds{E}[x]\mathds{E}[\nu]^\top\!\!L^\top\!\!B^\top\!\!\\
	&- 2\frac{1-p}{p}[BF_1\left(\mathds{E}[\tilde{x}\tilde{x}^\top\!\!]+\mathds{E}[x]\mathds{E}[x]^\top\!\!\right)\overline{F}^\top\!\!B^\top\!\! 
	\end{align*}
	where we have dropped the dependence on $k$ to increase the readability and denoted $\tilde{x} = \tilde{x}(k)$ and $\tilde{x}_+ = \tilde{x}(k+1)$.
	
	We want to find $F_1$ to minimise \eqref{eq:covariance_min}, therefore we compute
	\begin{align*}
	\frac{\partial Tr\left[\tilde{x}_+\tilde{x}^\top\!\!_+\right]}{\partial BF_1} =&
	2\frac{1-p}{p}BF_1\mathds{E}[xx^\top\!\!] + 2(1-p)BL\mathds{E}[\nu]\mathds{E}[x]^\top\!\! \\&- 2\frac{1-p}{p}B\overline{F}\mathds{E}[xx^\top\!\!]
	\end{align*}
and conclude that matrix $F_1$ in~\eqref{eq:F1} minimises variance \eqref{eq:covariance_min}. Given $F_1$, matrices $F_2$ and $L$ descend from~\eqref{eq:F2} and \eqref{eq:L}.
\end{proof}
	
The time-varying matrix gain in \eqref{eq:F1} should be used in the controller, but its computation is problematic because it requires the knowledge of expectation and variance of $\nu$. Therefore, we propose an approximation.

\subsection{Approximate covariance minimization controller}

The controller gain in \eqref{eq:F1} is time varying: to minimise the covariance of the error $x(k+1)$ the controller needs to know the expected value of the input at time $k$, i.e. $\mathds{E}[\nu(k)]$, the expected value of the error itself $\mathds{E}[x(k)]$ and its covariance $\mathds{E}[\tilde{x}(k)\tilde{x}(k)^\top\!\!]$. 
In order to derive a handier relation, we look for a static matrix to approximate the time varying $F_1$ in \eqref{eq:F1}. This approximation will be derived by looking at expected dynamics.
Even though $F_1$ is time-varying, the average controlled dynamics is time-invariant, namely
\begin{align}\label{eq:controlled_average_dynamics}
\left\{\begin{array}{rll}
\mathds{E}[x(k+1)] 	&= A\mathds{E}[x(k)] + B\mathds{E}[\xi(k)]+ E\mathds{E}[\nu(k)]\\
\mathds{E}[\xi(k)]	&=\overline{F}\mathds{E}[x(k)]+\overline{L}\mathds{E}[\nu(k)]
\end{array}\right.
\end{align}
and the corresponding transfer function is
%
\begin{align*}
G_{\nu\rightarrow\xi}=\frac{\Xi(z)}{N(z)}=\overline{F}(zI-A-B\overline{F})^{-1}(E+B\overline{L})+\overline{L}.
\end{align*}
%
%
%
Fig.~\ref{fig:u_s} shows by an example that this relation is effective in reconstructing the inputs, up to a delay that was not considered in the model. 
\begin{figure}\centering
	\includegraphics[width=.9\columnwidth]{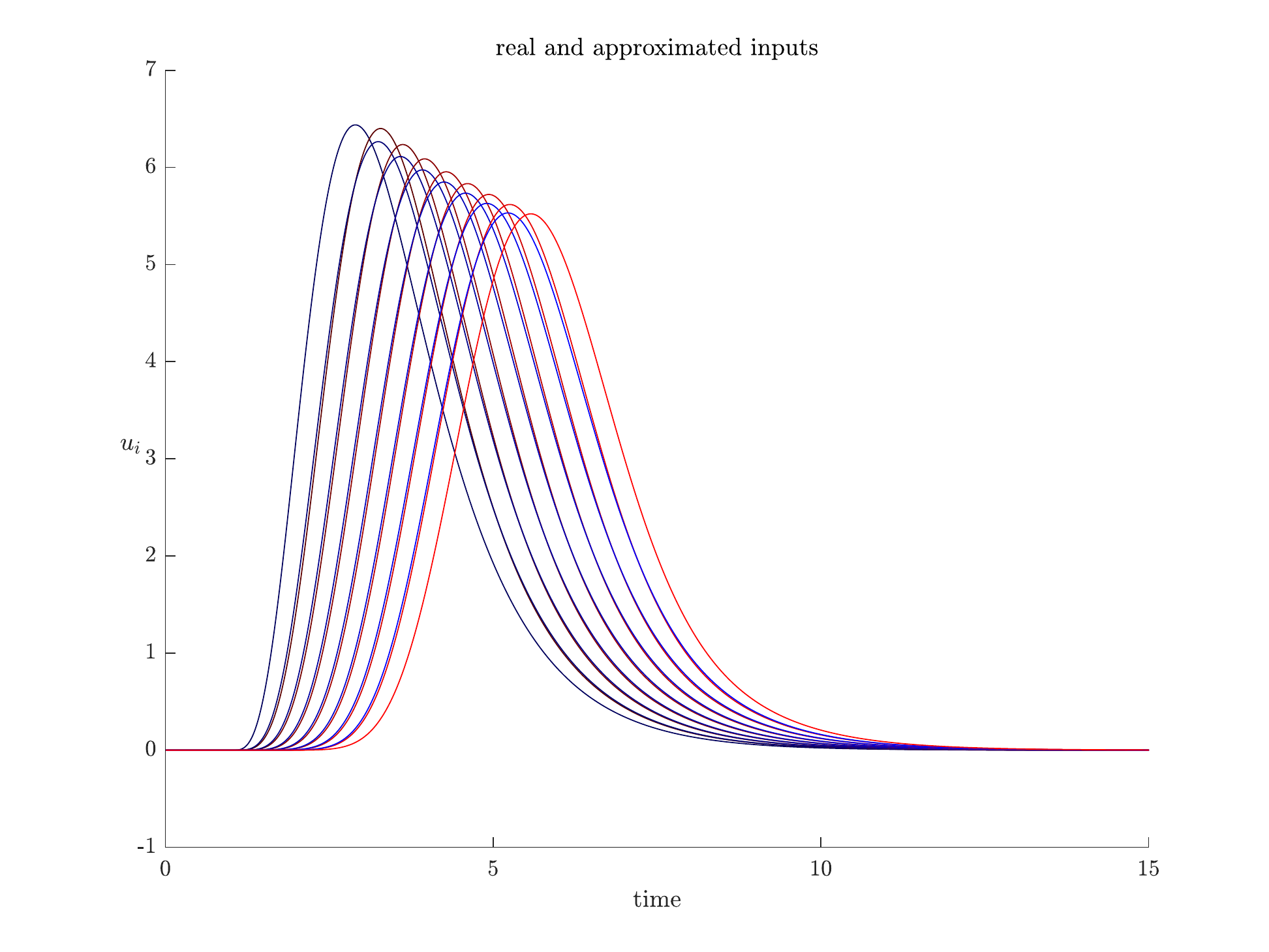}
	\caption{Real and estimated inputs for a platoon of 8 vehicles. The real inputs are displayed in red, while the output of $G_{\nu\rightarrow\xi}$ in blue. See Section~\ref{sect:simulations} for more information on the simulations.}
	\label{fig:u_s}
\end{figure}
Next, we further simplify the relation $\nu\rightarrow\xi$ by approximating it by its mere low frequency gain 
$g := \lim_{z\to 1}G_{\nu\rightarrow\xi}$
and therefore we approximate:
\begin{align}\label{eq:nu_approximation}
\mathds{E}[\nu(k)] \approx \frac{1}{g}\mathds{E}[\xi(k)].
\end{align}
By using the approximation \eqref{eq:nu_approximation} in \eqref{eq:controlled_average_dynamics} we get
\begin{align*}
\mathds{E}[\nu(k)] \approx \frac{1}{g }\mathds{E}[\xi(k)] = \frac{1}{g }\overline{F}\mathds{E}[x(k)] + \frac{1}{g }\overline{L}\mathds{E}[\nu(k)]
\end{align*}
and finally 
\begin{align}\label{eq:nu_approx_Ex}
\mathds{E}[\nu(k)] \approx \left(1-\frac{1}{g }\overline{L}\right)\frac{1}{g }\overline{F}\mathds{E}[x(k)].
\end{align}
Now we can use approximation~\eqref{eq:nu_approx_Ex} in \eqref{eq:F1}, which becomes
\begin{align*}
F_1 
\approx
&\overline{F}-\frac{p}{1-p}\overline{L}\left(1-\frac{1}{g }\overline{L}\right)\frac{1}{g }\overline{F}\mathds{E}[x]\mathds{E}[x]^\top\!\!\mathds{E}[xx^\top\!\!]^{-1}
\end{align*}
%
Finally, we disregard the statistical dispersion by approximating $\mathds{E}[x]\mathds{E}[x]^\top\!\!\mathds{E}[xx^\top\!\!]^{-1}$ $F_1$ by the identity matrix. We thus obtain
 \begin{align}\label{eq:F1approx}
F_1 \approx \left(1-\frac{p}{1-p}\overline{L}\left(1-\frac{1}{g }\overline{L}\right)\frac{1}{g }\right)\overline{F}
\end{align}
which is the constant gain used in our implementation.


\subsection{State observer design}
Even though in Section~\ref{sect:full-info-controller} we have designed a state-feedback controller, the output of system~\eqref{eq:system_control} is not the full state. Since one input is unknown, we need to estimate the state, which we do via an Unknown Input Observer~\cite{Frank:1992aa,Popescu:2018aa}. Because of the measurement delay, we can only estimate the delayed state
	\begin{align}
		x_d(k) = x(k-m) \label{eq:delayed_state} 
	\end{align}
	and by substituting \eqref{eq:delayed_state} in \eqref{eq:system_control} we get:
	\begin{align*}
	x_d(k+1) 	&= Ax_d(k) + B\xi_d(k) +E\nu_d(k) \\
	y(k) 		&= Cx_d(k)  \nonumber
	\end{align*}
	where $\xi_d(k) = \xi(k-d-m)$ and $\nu_d(k) = \nu(k-d-m)$ are the delayed inputs. 
		
	As observer for the delayed state $x_d(k)$  it is possible to use the following dynamical system:
	\begin{align}
	\zeta(k+1) &= F\zeta(k) + GB\xi_d(k) +Ky(k) \label{eq:UIO} \\
	\hat{x}(k) &=\zeta(k) + Hy(k) \nonumber
	\end{align}
	After simple algebraic manipulations, it can be shown that the estimate error $\varepsilon(k) = x_d(k) - \hat{x}(k)$ follows the dynamics:
	\begin{align*}
	\varepsilon(k+&1) 
	=Ax(k) + B\xi_d(k) + E\nu_d(k)-F\zeta(k)-GB\xi_d(k)\\&-K_1Cx(k)-K_2y(k)-Hy(k+1)\\
	=&Ax(k)+B\xi_d(k)+E\nu_d(k)-F(\hat{x}(k)-Hy(k))\\&-GB\xi_d(k)-K_1Cx(k)-K_2y(k)\\
	&-HCAx(k)-HCB\xi_d(k) - HCE\nu_d(k)\\
	=&(A-K_1C-HCA)(x(k)-\hat{x}(k)) + (A-K_1C-HCA-F)\hat{x}(k)\\
	&+(I-G-HC)B\xi_d(k) +(I-HC)E\nu_d(k) + (FH-K_2)y(k) 
	\end{align*}
	where $K = K_1+K_2$. Since $CE$ is injective, we can choose
	\begin{align*}
	F &= A - K_1C - HCA\\
	K_2 &= FH\\
	H &= E[(CE)^\top\!\!CE]^{-1}(CE)^\top\!\!\\
	G &= I - HC
	\end{align*}
	to obtain
	\begin{align*}
	\varepsilon(k+1) = (A-K_1C-HCA) \varepsilon(k)
	\end{align*}
In order to bring the estimate error to zero, we can choose $K_1$ to have the desired poles for $F$, because the couple $(A-HCA, C_1)$ is observable. A straightforward solution is choosing the gain matrix $K_1$ to impose a deadbeat response with all the poles of $F$ in $0$. 
In our implementation, we use $\hat x(k)$ as estimate for $x(k)$, effectively disregarding the measurement delay.

\subsection{Modeling input delays: Lifted system}
So far, our design has disregarded delays. Indeed, delays can be accounted for by writing a standard lifted system.
Let us recallf that the system to control, for the full information case and including the input delay,  is 
\begin{align}\label{eq:delay_fullinfo}
\left\{\begin{array}{rll}
x(k+1) 	&= Ax(k) + B\xi(k-d) + E\nu(k)\\
z(k) 	&=C_zx(k) + R\xi(k)
\end{array}\right.
\end{align}
where $C_z = \left[\begin{array}{ccc}\varepsilon &0 &0\\0&0&0\end{array}\right]$ and $R = \left[\begin{array}{c}0\\1\end{array}\right]$. 
By defining
\begin{align*}
\xi_e(k) &= \left[\begin{array}{cccc}\xi(k-d+1)&\xi(k-d+2)&\dots&\xi(k-1)\end{array}\right]^\top\!\!\\
\nu_e(k) &= \left[\begin{array}{cccc}\nu(k-d+1)&\nu(k-d+2)&\dots&\nu(k-1)\end{array}\right]^\top\!\!\\
x_e(k) &=  \left[\begin{array}{ccccc}x(k)^\top\!\! &\xi(k-d) &\xi_e(k)^\top\!\!&\nu(k-d)&\nu_e(k)^\top\!\!\end{array}\right]^\top\!\!
\end{align*}
we can write~\eqref{eq:delay_fullinfo} as
\begin{align*}
\left[\begin{array}{c}x(k+1)\\\xi(k-d+1)\\\xi_e(k+1)\\\nu(k-d+1)\\\nu_e(k+1)\end{array}\right]
=&
\left[\begin{matrix}
A&B&0&E&0\\
0&0&e_1^\top\!\!&0&0\\
0&0&\Omega&0&0\\
0&0&0&0&e_1^\top\!\!\\
0&0&0&0&\Omega
\end{matrix}\right]
\left[\begin{array}{c}x(k)\\\xi(k-d)\\\xi_e(k)\\\nu(k-d)\\\nu_e(k)\end{array}\right] \\&+
\left[\begin{array}{c}0_{n,n}\\0\\e_{d}\\0\\0_d\end{array}\right]\xi(k)+
\left[\begin{array}{c}0_{n,n}\\0\\0_d\\0\\e_d\end{array}\right]\nu(k)\\
z(k)&=		\left[\begin{array}{ccc}C_z&0_{2,d}^\top\!\!&0_{2,d}^\top\!\!\end{array}\right]x_e(k)	+ R \xi(k)
\end{align*}
where
$e_1 = \left[\begin{array}{cccc}1&0&\dots&0\end{array}\right]^\top\!\!$, $e_d = \left[\begin{array}{cccc}0&\dots&0&1\end{array}\right]^\top\!\!$ and
\begin{align*}
\Omega = \left[\begin{matrix}
0&1&0&0&\dots&0\\
0&0&1&0&\dots&0\\
\vdots&&&\ddots&&\vdots\\
0&0&0&0&1&0\\
0&0&0&0&0&1\\
0&0&0&0&0&0
\end{matrix}\right] \in \mathds{R}^d\times\mathds{R}^d
\end{align*}
This dynamics can be summarized in the form
\begin{align}\label{eq:lift_control_sys}
\left\{\begin{array}{rl}
x_e(k+1) &= A_dx_e(k)+B_d\xi(k)+E_d\nu(k)\\
z(k) &= C_dx(k)+R\xi(k),
\end{array}\right.
\end{align}
which directly parallels the form~\eqref{eq:system_control}.
The design described in the previous sections on system~\eqref{eq:system_control} shall actually be applied to this lifted system to compute a controller that takes into account the delays.

%
	\begin{figure*} 
		\centering
		\includegraphics[width=.95\textwidth]{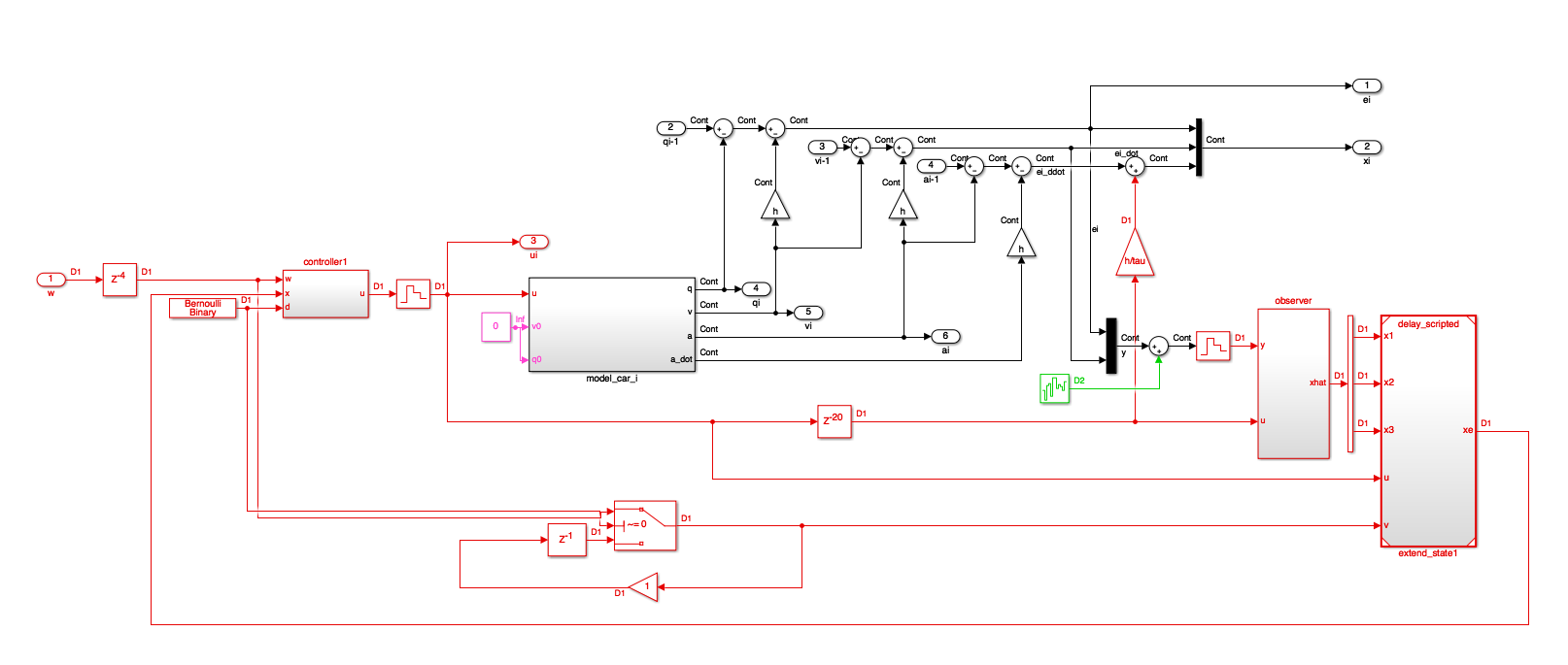}
		\caption{Comprehensive single vehicle Simulink model. The $\mathcal{H}_\infty$ controller is \emph{controller1}, the unknown input observer is \emph{observer} and the lifting block is \emph{extend\_state1}.
		The continuous dynamics~\eqref{eq:dynamics} is simulated in \emph{model\_car\_i}.}
	  \label{fig:single_car_ss}
	\end{figure*}
\section{Numerical results}\label{sect:simulations}
We have implemented our dynamical model and our controller using the software Matlab/Simulink. 
Each vehicle and its controller are implemented as in Fig.~\ref{fig:single_car_ss}: the implementation includes features that have been disregarded during the design, such as transmission delays, input delays and measurement noise.
The vehicle dynamics is simulated in continuous time~\eqref{eq:dynamics}, while the controller is digital as described in Section~\ref{sect:control-design}. The solver used is \texttt{ode23s}, as the large number of delays in the system leads to chattering if a non-stiff solver 
is used.
The dynamical parameters are set as $\tau = 0.1s$, input delay $\phi = 0.2s$, transmission delay $\theta = 0.02s$, and measurement delay $\psi = 0.05s$.  We also have observed robustness to small measurement noise (up to $1-2\%$). The sampling rate used by the digital controller and the radio is $T_s = 0.01s$. We set the time-headway $h=0.25$.
The deterministic $\mathcal{H}_\infty$ problem is defined with $\eps=10^{-1}$ and $r=1$, and it is solved, exploiting the continuous-time equivalence presented in \cite{Chen:1994aa} by using the continuous Riccati equation solver in Matlab. In the covariance minimization, the approximation \eqref{eq:F1approx} is used with $g=0.9734$. After lifting, the system has order 43: this figure is relatively large but compatible with the solvers we use for design.  
Each packet transmitted by one vehicle to the following one is received with probability $1-p$.
	
In order to test the ability of our controller to ensure string stability, we simulate the following dynamic scenario with time-headway $h=0.25$. 	At time zero, all vehicles start at rest. The first vehicle of the platoon is connected to a \emph{virtual} leader vehicle: this is a Simulink block that creates the trajectory to be followed by the platoon, linearly ramping up in speed from 0 to $17m/s$, then keeping this speed constant.

In this dynamic scenario, we first test our controller in the case when there are no losses ($p=0$).	
Fig.~\ref{fig:all_noloss_3} illustrates the string stable behavior that is obtained in this case: observe how the input to each vehicle becomes monotonically smaller in the downstream direction. For completeness, Fig.~\ref{fig:all_noloss_3} also shows the state variables position, speed, and acceleration.
Our tests show that the controller can handle headways at least as small as 0.2, which is consistent with the state of the art~\cite{Ploeg:2015}. 
	\begin{figure*}
		\centering
			\includegraphics[width=.96\columnwidth]{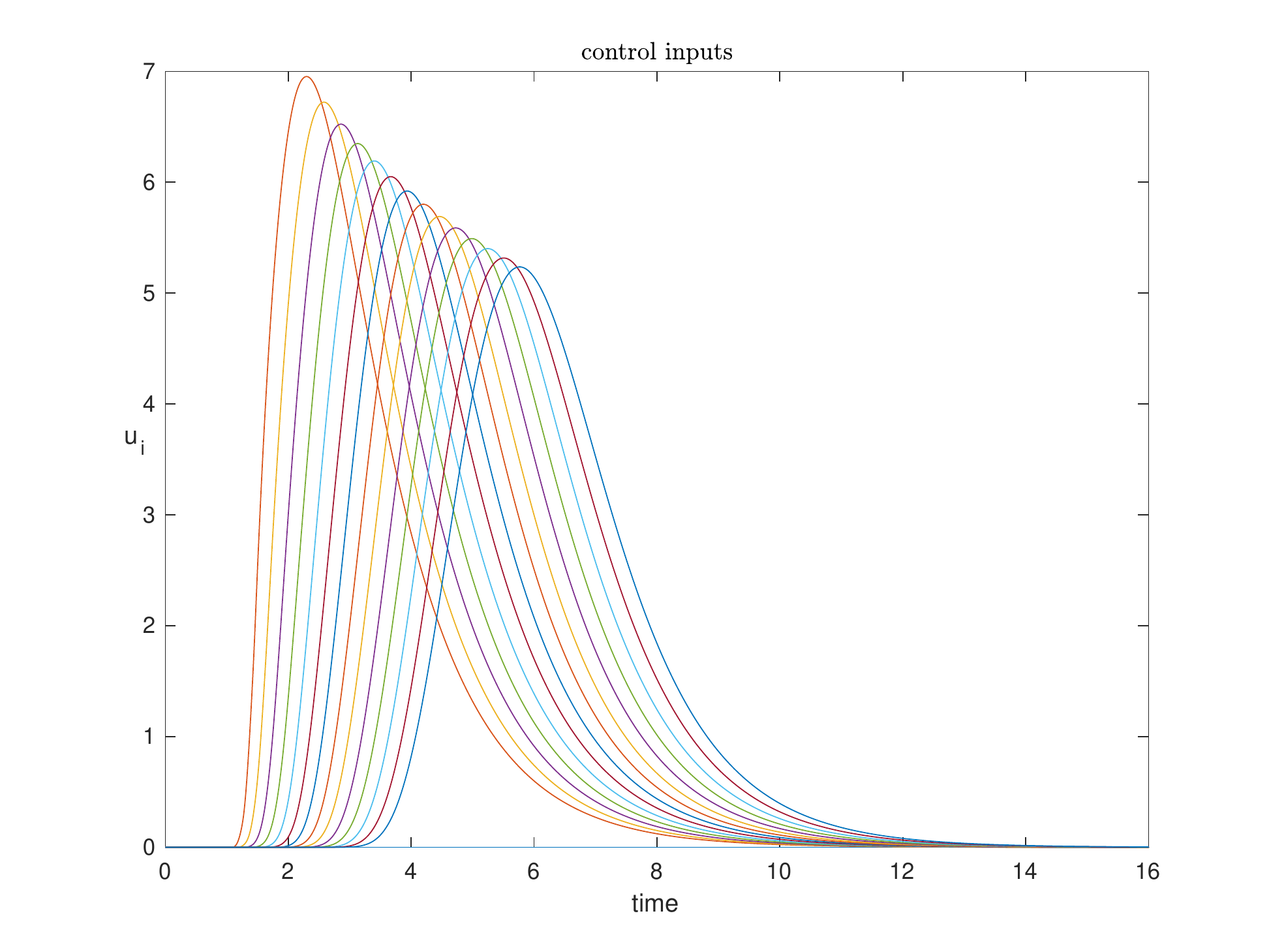}
			\includegraphics[width=.96\columnwidth]{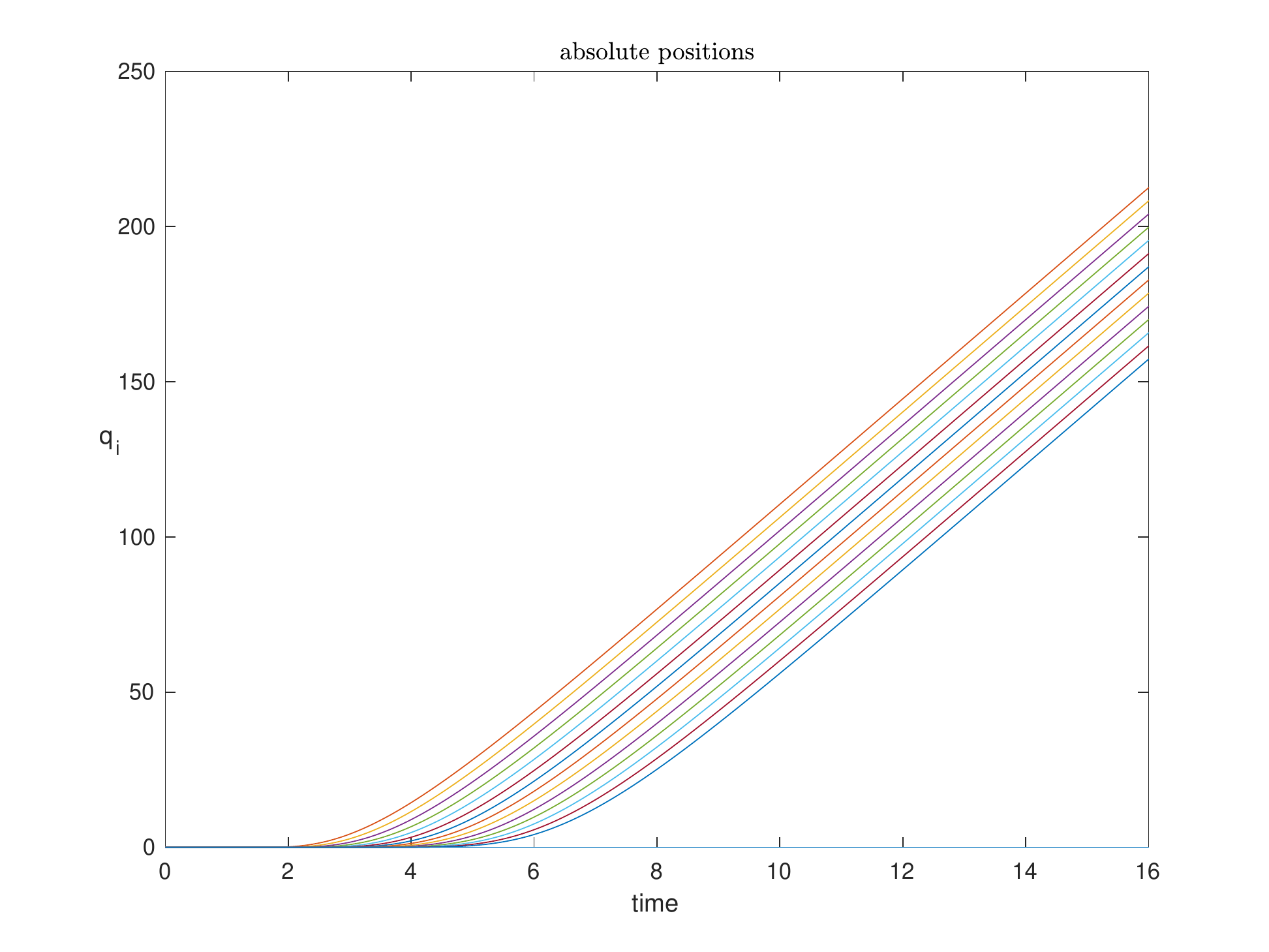}
			\includegraphics[width=.96\columnwidth]{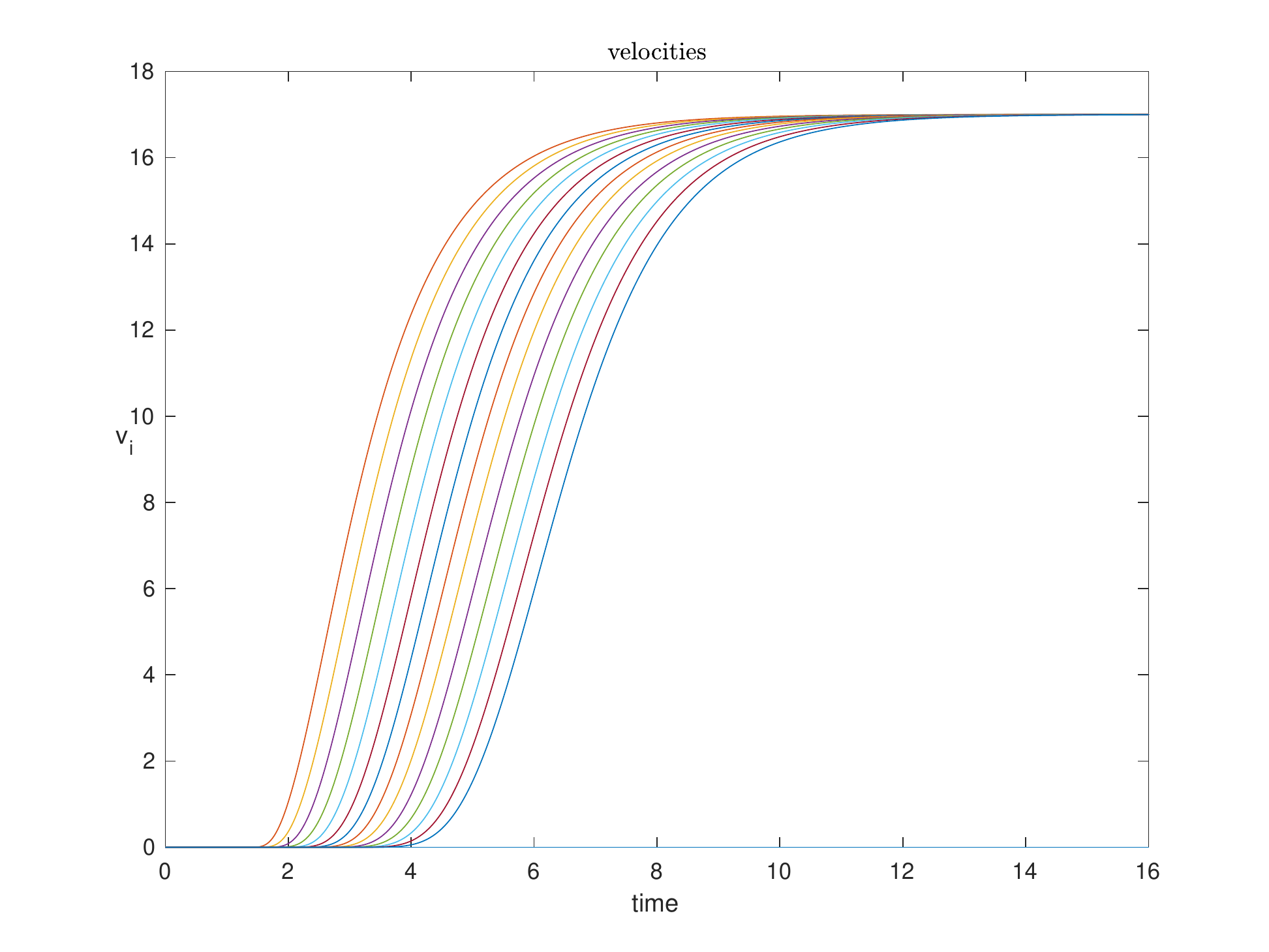}
			\includegraphics[width=.96\columnwidth]{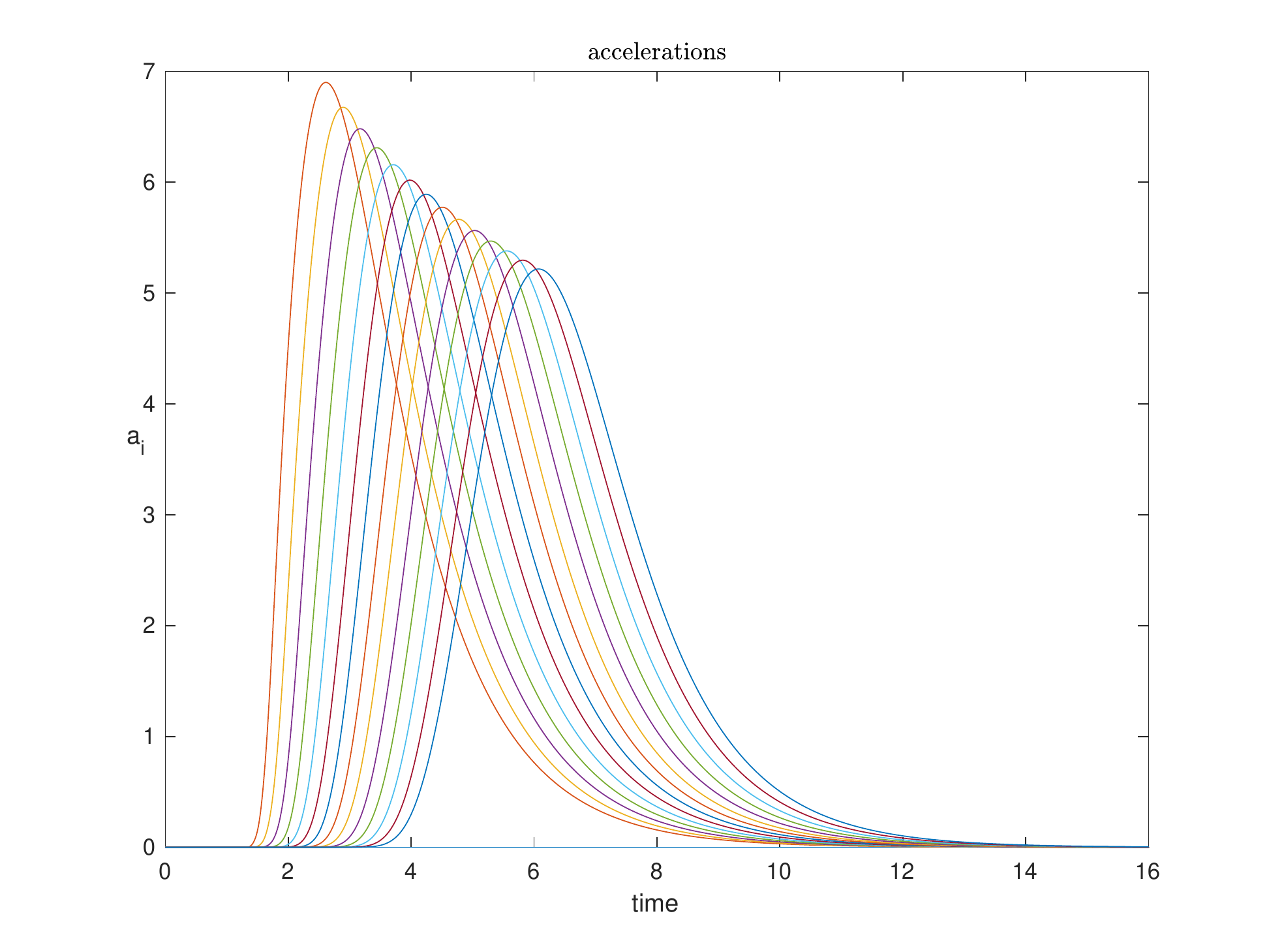}
		\caption{Time evolutions of inputs, positions, velocities and accelerations for a controlled platoon of 14 vehicles when no losses are present.}\label{fig:all_noloss_3}
	\end{figure*}

After evaluating the performance when the platoon experiences no losses, we are interested in the behaviour when each packet has a positive probability $p$ of not being received.	
The switching controller designed in the previous section shows very good ability to stabilize the platoon in the presence of losses, even for very large loss probabilities. The control inputs for a simulation with $p=0.8$ are displayed in Fig.~\ref{fig:u_p8_h25}, showing very good string stability. Performance only begins to deteriorate for $p=0.9$, as can be seen in Fig.~\ref{fig:u_p9_h25_bar}.

We want to stress that our switching architecture with covariance minimization is indeed key to achieve these good results.
%
To make this fact apparent, we have performed some simulations using controller~\eqref{eq:non-switching_control_law} where the nominal gain matrix matrices $\overline{F}$ and $\overline{L}$ are constantly used and missing values of $\nu$ are replaced by the most recent value received in the past.
In this case, except for the random delays that are induced by the packet losses, the resulting closed-loop systems would be string stable in expectation.
This property might let us hope for a good performance, but in fact the realized trajectories are far from being nicely string stable. The improvements brought by the covariance minimizing controller become evident by comparing Fig.~\ref{fig:u_p8_h25} against Fig.~\ref{fig:u_p8_h25_bar} or Fig.~\ref{fig:u_p9_h25} against Fig.~\ref{fig:u_p9_h25_bar}, respectively.

%
%
	
\begin{figure*}
		\centering
		\begin{subfigure}[t]{0.96\columnwidth}			
			\includegraphics[width=\columnwidth]{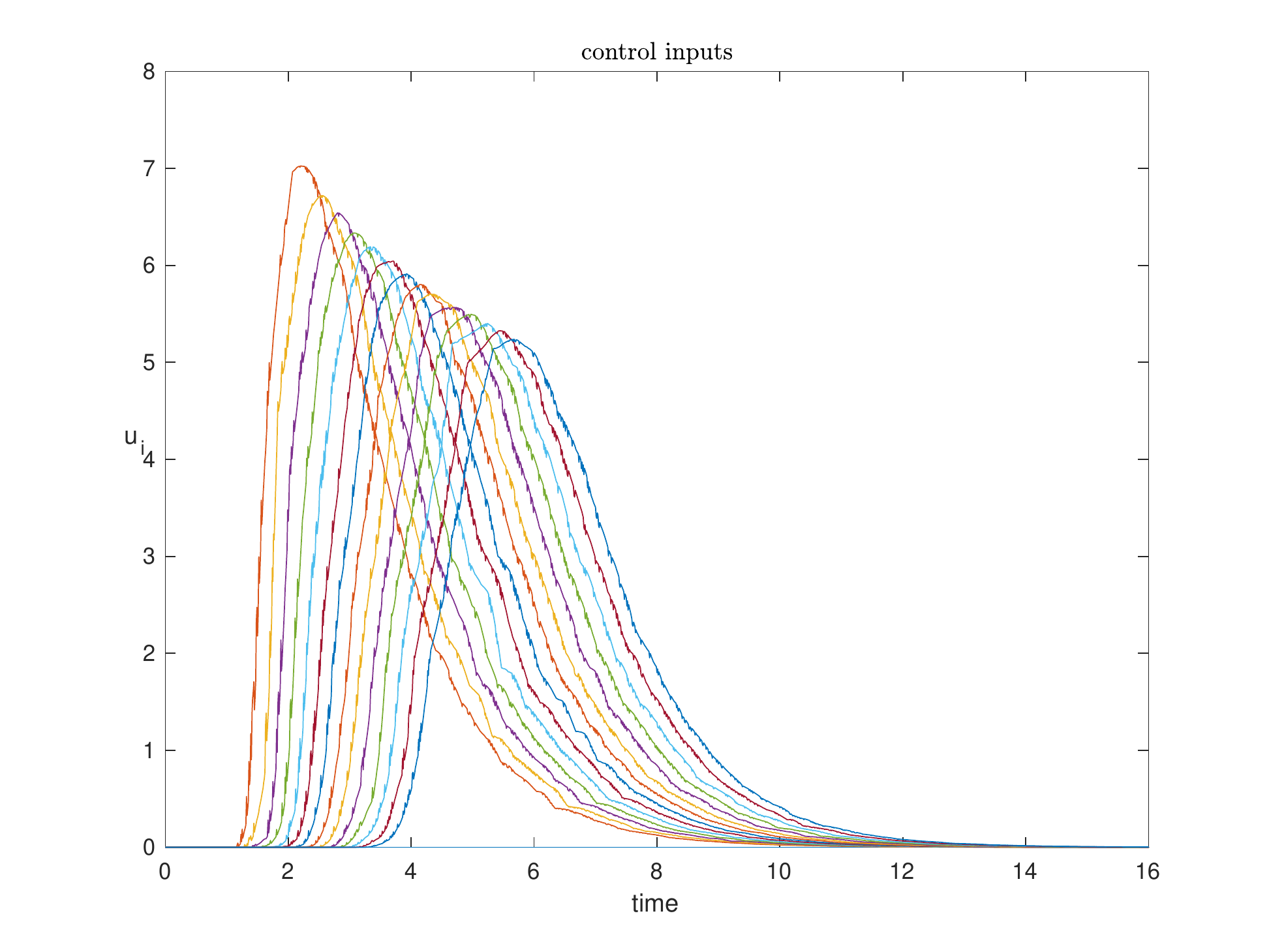}
			\caption{Minimizing covariance controller\label{fig:u_p8_h25_bar}}
		\end{subfigure}
		\begin{subfigure}[t]{.96\columnwidth}
			\includegraphics[width=\columnwidth]{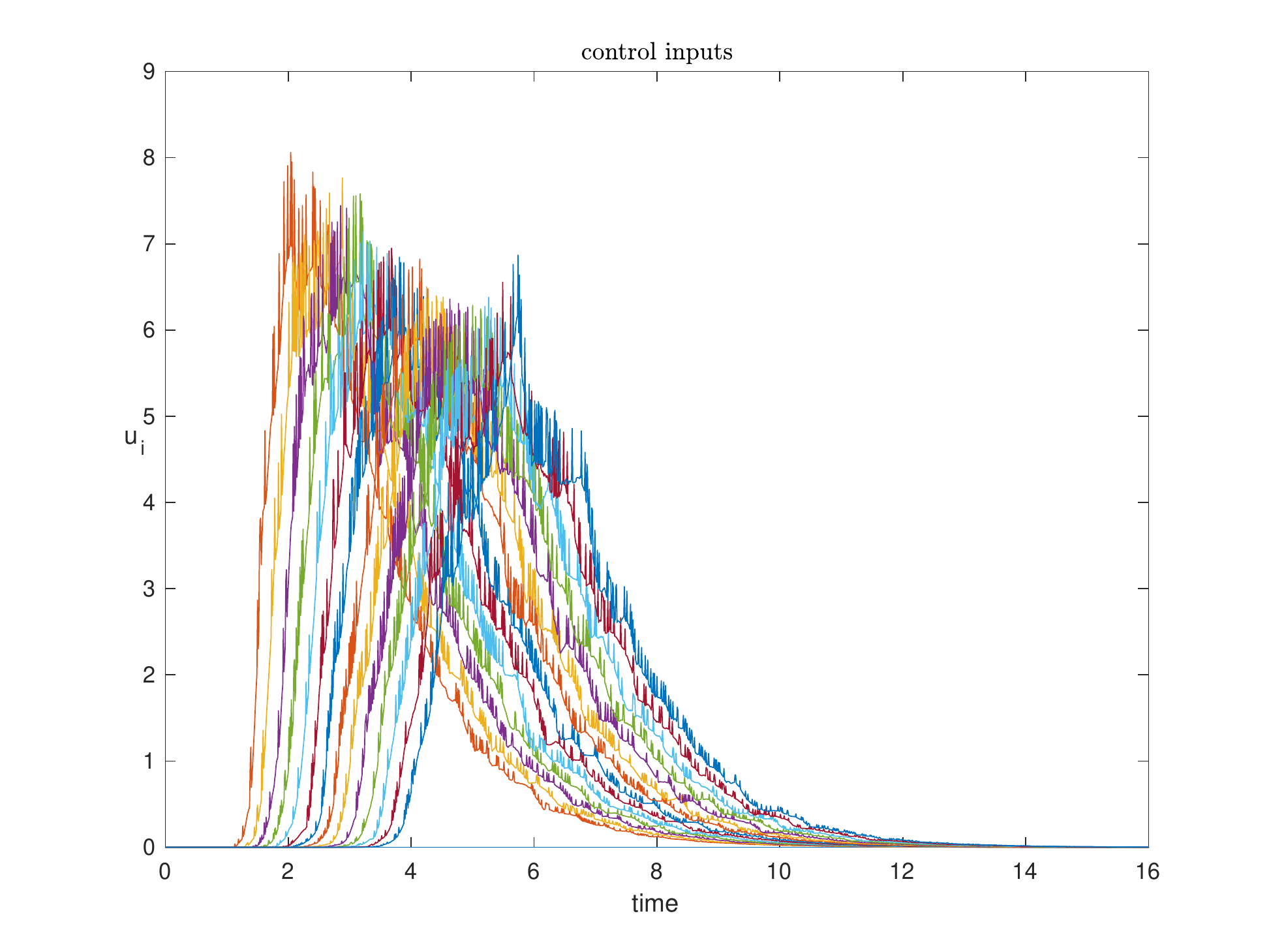}
			\caption{Non-switching controller	\label{fig:u_p8_h25}}
		\end{subfigure}
		\caption{Control inputs for a platoon of 14 vehicles with $p = 0.8$.}
	\end{figure*}
	\begin{figure*}[h]
		\centering
		\begin{subfigure}[t]{.96\columnwidth}
			\centering
			\includegraphics[width=\columnwidth]{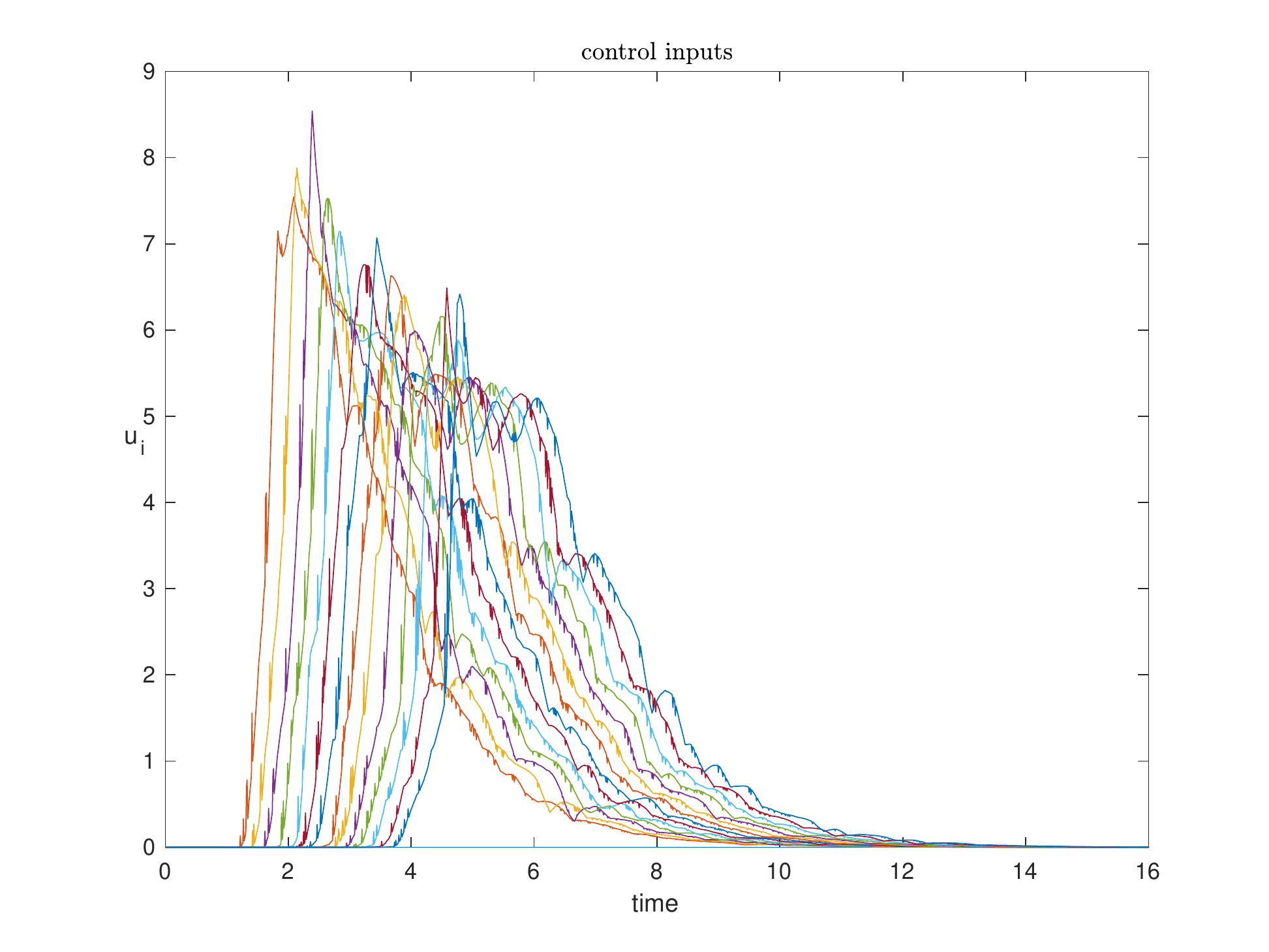}
			\caption{Minimizing covariance controller	\label{fig:u_p9_h25_bar}}
		\end{subfigure}	
		\begin{subfigure}[t]{.96\columnwidth}
			\centering
			\includegraphics[width=\columnwidth]{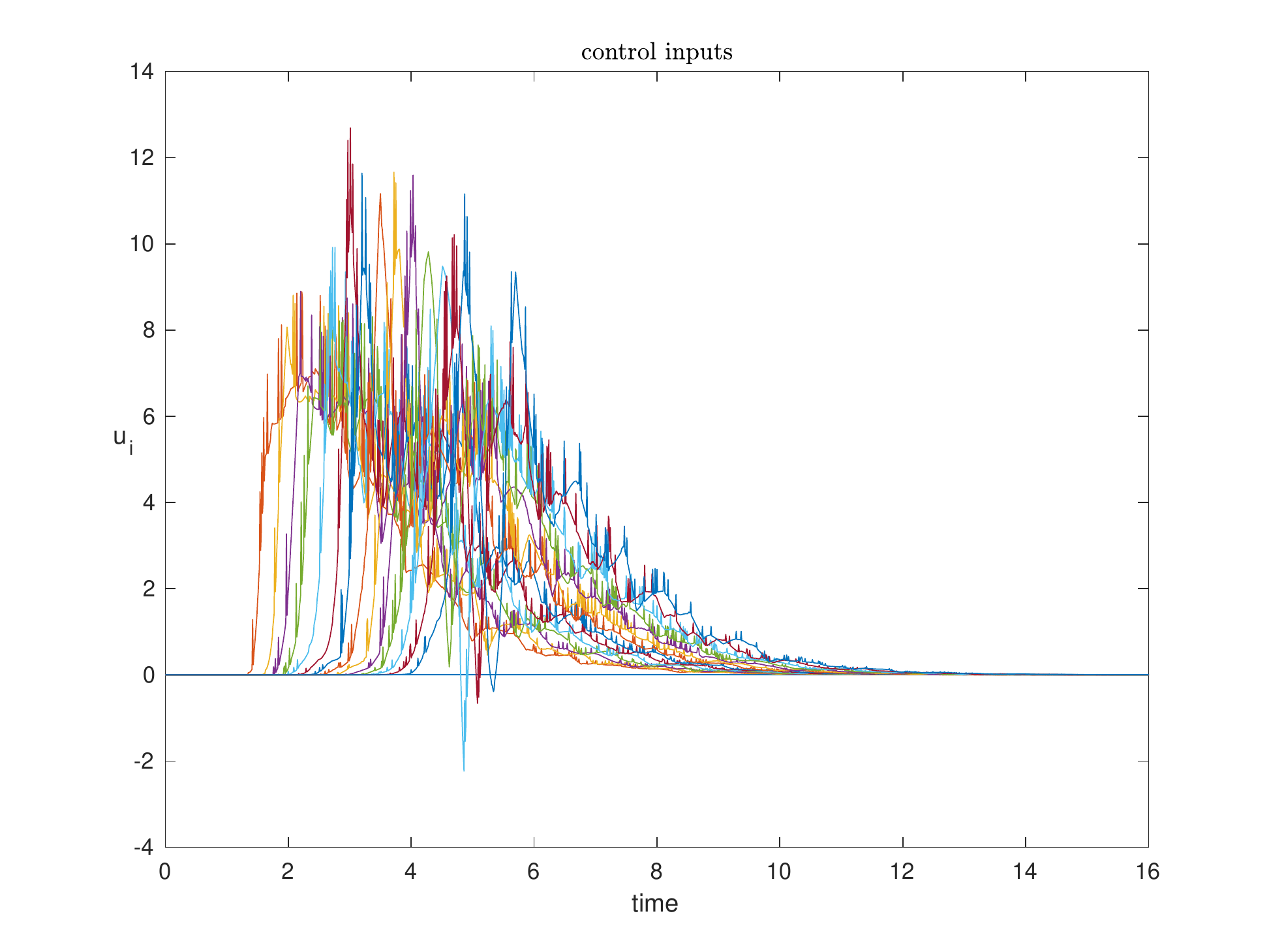}
			\caption{Non-switching controller	\label{fig:u_p9_h25}}
		\end{subfigure}
		\caption{Control inputs for a platoon of 14 vehicles with $p = 0.9$.}
	\end{figure*}

\section{Conclusion}\label{sect:outro}

In this paper, we have addressed the question of designing a distributed cooperative controller that can robustly stabilize a platoon of vehicles when the communication is affected by losses. This has been achieved by a switching controller that has been designed with a twofold objective: promoting plant stability and string stability of the average dynamics, and promoting trajectories to be close to the average. The former objective has been sought by applying $\mathcal{H}_\infty$ control tools to the average dynamics, while the latter has been sought by minimizing the variance of the trajectories.

Even though the $\mathcal{H}_\infty$ approach of minimizing the ratio~\eqref{eq:H-inf} does not (strictly speaking) guarantee that the string stability condition \eqref{eq:z-stability} is met for a given value of $h$, minimizing the ratio makes it less than one when there exists a controller in the form 
		$\xi(k) = \overline{F}x(k) + \overline{L}\nu(k)$ that is 
capable of achieving string stability. In other words, we can say that if the system without losses can be made string stable by a controller, then our controller will make it string stable.

A key point in our work has been the stochasticity of the communication and the consequent need to cope with it. In our design solution, we have chosen to focus on the first two moments: the first moment to ensure a good average behavior and the second moment to make such good behavior likely.  
In designing the control for the average, the choice of combining state feedback with an unknown input observer has two advantages. First, it keeps the $\mathcal{H}_\infty$ problem tractable; second, it provides an avenue to refine our control design. Indeed, the unknown input observer can be used to produce an estimate of the input $\nu$, which could then replace the communicated value when the latter is unavailable. This design option has been tested in a frequency domain approach in~\cite{Acciani:2018ab} with positive results. 

Thanks to minimizing the variance among the trajectories, our controller shows very good ability to cope with high levels of losses without degrading the platoon performance (in terms of the headway $h$). This features constitutes an improvement upon previous works that aimed at ``graceful degradation'' of performance in presence of losses~\cite{Ploeg:2015}.
In order to refine our way to deal with randomness, probabilistic guarantees on the set of possible trajectories would be the natural objective of future work.  In the context of analysis, relevant definitions and results have been recently given by \cite{Qin:2017aa}, with the notion of $n\sigma$ string stability (that is, that all trajectories within a neighborhood of radius $n$ times the standard deviation about the average are string stable). 

%

\ifCLASSOPTIONcaptionsoff
  \newpage
\fi



\bibliographystyle{IEEEtran}
\bibliography{../general_biblio}%

%








\end{document}